\documentclass{article}
\usepackage[T1]{fontenc}
\usepackage{amsfonts,amsmath,amsthm,amssymb,authblk,latexsym,bm}
\usepackage{hyperref} \hypersetup{colorlinks=true}
\usepackage{citeref}
\usepackage[cmyk]{xcolor} 

\newtheorem{theorem}{Theorem}[section]

\newtheorem{lemma}[theorem]{Lemma}

\newtheorem{example}[theorem]{Example}

\newtheorem{remark}[theorem]{Remark}

\numberwithin{equation}{section}

\def\CR{{\cal R}} \def\CA{{\cal A}} \def\CD{{\cal D}}
\def\Z{{\Bbb Z}}

\begin{document}

\title{Self-dual 2-quasi Negacyclic Codes over Finite Fields}
\author{
Yun Fan, ~ Yue Leng\par
{\small School of Mathematics and Statistics}\par\vskip-1mm
{\small Central China Normal University, Wuhan 430079, China}
}

\date{}
\maketitle

\insert\footins{\footnotesize{\it Email address}:
yfan@mail.ccnu.edu.cn (Yun Fan);

}

\begin{abstract}
In this paper,
we investigate the existence and asymptotic property
of self-dual $2$-quasi negacyclic codes of length $2n$
over a finite field of cardinality $q$.
When $n$ is odd, we show that the $q$-ary
self-dual $2$-quasi negacyclic codes exist
if and only if  $q\,{\not\equiv}-\!1~({\rm mod}~4)$.
When $n$ is even,
we prove that the $q$-ary self-dual $2$-quasi negacyclic codes always exist.
By using the technique introduced in this paper,
we prove that $q$-ary self-dual $2$-quasi negacyclic codes
are asymptotically good.

\end{abstract}

\section{Introduction}

Let $F$ be a finite field with cardinality $|F|=q$, where $q$ is a prime power.
Any ${\bf a}=(a_0,a_1,\dots,a_{n-1})\in F^n$ is called a word over $F$
of length $n$, where $n$ is a positive integer. For ${\bf a},{\bf b}\in F^n$
the inner product $\langle{\bf a},{\bf b}\rangle =\sum_{i=0}^{n-1} a_ib_i$.
Any linear subspace $C$ of $F^n$ is called a {\em linear code}.
The fraction ${\rm R}(C)=\frac{\dim _F C}{n}$ is called the {\em rate} of $C$.
The $C^\bot
=\{{\bf a}\in F^n\,|\,\langle{\bf c},{\bf a}\rangle=0,\;\forall\,{\bf c}\in C\}$
is called the {\em dual code} of $C$.
And $C$ is said to be {\em self-dual} if~$C=C^\bot$.
It is known that ${\rm R}(C)=\frac{1}{2}$ if~$C$ is self-dual.
The minimum weight
${\rm w}(C)=\min_{0\ne {\bf c}\in C}{\rm w}({\bf c})$,
where the weight ${\rm w}({\bf c})$ of ${\bf c}=(c_0,\dots,c_{n-1})$
is defined to be the number of the indexes $i$ that $c_i\ne 0$.
The larger both the {\em relative minimum weight}
$\Delta(C)=\frac{{\rm w}(C)}{n}$ and
the {\em rate} ${\rm R}(C)$ are,
the better the coding performance of the code $C$ is.
A sequence of codes $C_1, C_2, \dots$ is said to be
{\em asymptotically good} if the code length of $C_i$ goes to infinity
and there is a real number
$\delta>0$ such that ${\rm R}(C_i)>\delta$ and $\Delta(C_i)>\delta$ for
$i=1,2,\dots$.
A class of codes is said to be {\em asymptotically good} if there is an
asymptotically good sequence of codes in the class.

It is known since long time that
the linear codes over $F$ are asymptotically good (see \cite{G52, V57}).
Let the real number $\delta\in(0,1-q^{-1})$,
then there is a sequence of linear codes
$C_1,C_2,\dots$ over $F$ with the code length going to infinity
such that $\Delta(C_i)>\delta$ and the rate
${\rm R}(C_i)$ is approximately $1-h_q(\delta)$,
where $h_q(\delta)$ is the {\em $q$-entropy function} as follows
\begin{align} \label{eq def h_q}
h_{q}(\delta)=
\delta \log_{q}(q\!-\!1)-\delta \log_{q}(\delta)-(1\!-\!\delta)\log_q(1\!-\!\delta),
 ~~~
 \delta\in[0,1\!-\!q^{-1}],
\end{align}
which value strictly increases from $0$ to $1$
while $\delta$ increases from $0$ to $1-q^{-1}$.
And, $1-h_q(\delta)$ is called the {\em GV-bound}.

A linear code $C$ of $F^n$ is called a cyclic code
if for any $(c_0,c_1,\dots,c_{n-1})$ of~$C$
the cyclically permuted word $(c_{n-1},c_0,\dots,c_{n-2})$ still belongs to $C$.
Cyclic codes are a vital class of codes (e.g., cf. \cite{HP}). However, it is a long-standing
open question (e.g., cf. \cite{MW06}): are cyclic codes over $F$ asymptotically good?
Alternatively, a variation class called
{\em quasi-cyclic codes of index $2$} are proved asymptotically good.
A linear code $C$ of $F^n\times F^n$
is called a quasi-cyclic code of index $2$,
{\em $2$-quasi cyclic code} for short,
if the following holds:
\begin{align*}
& (c_0,c_1,\dots,c_{n-1},\; c_0',c_1',\dots,c_{n-1}')\in C \\
&\implies  (c_{n-1},c_0,\dots,c_{n-2}, \; c_{n-1}',c_0',\dots,c_{n-2}')\in C.
\end{align*}
The binary $2$-quasi cyclic codes were proved asymptotically good in
\cite{CPW, C, K}. And, it was shown in \cite{L PhD}
(or cf. \cite[Theorem III.12]{FL22}) that
any $q$-ary $2$-quasi cyclic codes are asymptotically good and
attain the GV-bound.
Even more, binary self-dual $2$-quasi cyclic codes are asymptotically good,
see \cite{MW}. And,
if $q\,{\not\equiv}\,-1~({\rm mod}~4)$,
then the $q$-ary self-dual $2$-quasi cyclic codes are asymptotically good;
precisely, there are $q$-ary self-dual $2$-quasi cyclic codes $C_1,C_2,\dots$
with code length going to infinity such that
$\Delta(C_i)>\delta$,
where $h_q(\delta)<\frac{1}{4}$ is lower than the GV-bound
(note that ${\rm R}(C_i)=\frac{1}{2}$ since $C_i$ is self-dual);
see \cite{AOS, L PhD, LF22}.
The proof of \cite{AOS} is based on
the Artin's  primitive root conjecture;
while in \cite{L PhD, LF22} the conjecture is not required.
Note that ``$q\,{\not\equiv}-\!1~({\rm mod}~4)$''
is a necessary and sufficient condition for the existence of
$q$-ary self-dual $2$-quasi cyclic codes of length $2n$,
see \cite[Proposition 6.1]{LS01} for the semisimple case, 
and see \cite[Theorem~6.1]{LS03} for general case.

Let $0\neq\lambda\in F$.
A linear code $C$ of $F^n\times F^n$
is called a {\em $2$-quasi $\lambda$-constacyclic code} if the following holds:
\begin{align} \label{eq 2-Q lambda-CC}
\begin{array}{l}
 (c_0,c_1,\dots,c_{n-1},\; c_0',c_1',\dots,c_{n-1}')\in C \\
 \implies  (\lambda c_{n-1},c_0,\dots,c_{n-2}, \;
  \lambda c_{n-1}',c_0',\dots,c_{n-2}')\in C;
\end{array}
\end{align}
in particular, such $C$ is just a $2$-quasi cyclic code once $\lambda=1$.
If $\lambda=-1$ then such $C$ is called a {\em $2$-quasi negacyclic code}.
It has been shown that the $2$-quasi $\lambda$-constacyclic codes over $F$
are asymptotically good and attain the GV-bound, see \cite[Theorem III.12]{FL22}.
About the self-dual ones,
it is known 
that, if $\lambda\ne\pm 1$
then the $q$-ary self-dual $2$-quasi $\lambda$-constacyclic codes
exist if and only if $q\,{\not\equiv}-\!1~({\rm mod}~4)$,
and in that case the $q$-ary self-dual $2$-quasi $\lambda$-constacyclic codes
are not asymptotically good, see \cite[Corollary 4.3, Theorem 4.5]{FL24}.

As for the self-dual $2$-quasi negacyclic codes,
 their existence and their asymptotic property seem an interesting question.
When $q$ is even, they are just self-dual $2$-quasi cyclic codes.
In the case that $q\,{\equiv}-\!1~({\rm mod}~4)$,
Shi et al. \cite{SQS} showed, based on the Artin's primitive root conjecture, that
there are $q$-ary self-dual $2$-quasi negacyclic codes
$C_1,C_2,\dots$ with code length of $C_i$ being $4p_i$
(i.e., $n=2p_i$ in Eq.\eqref{eq 2-Q lambda-CC}) where $p_i$ are primes
going to infinity,
and $\Delta(C_i)>\delta$, where $h_q(\delta)<\frac{1}{8}$.
Under another assumption that $q$ is odd and
$q\,{\not\equiv}\pm 1~({\rm mod}~8)$,
Alahmadi et al. \cite{AGOSS} proved that
there are $q$-ary self-dual $2$-quasi negacyclic codes
$C_1,C_2,\dots$ with code length of $C_i$ being powers of $2$ and going to infinity,
and $\Delta(C_i)>\delta$, where $h_q(\delta)<\frac{1}{4}$.
A key technique in \cite{AGOSS} is that under their assumption on $q$
the $X^{2^m}+1$ is a product of two irreducible polynomials over~$F$
which are reciprocal each other.
About the existence, assuming that $n$ is odd and $\gcd(n,q)=1$,
they proved in \cite{AGOSS} that
the $q$-ary self-dual $2$-quasi negacyclic codes of length $2n$ exist if and only if
 $q\,{\not\equiv}-\!1~({\rm mod}~4)$.

We are concerned with the existence and the asymptotic property of
self-dual $2$-quasi negacyclic codes over $F$.
About the existence 
 we get the following result.

\begin{theorem} \label{int existence}
{\bf(1)} If $n$ is odd, then the $q$-ary self-dual $2$-quasi negacyclic codes
of length $2n$ exist if and only if  $q\,{\not\equiv}-\!1~({\rm mod}~4)$.

{\bf(2)} If $n$ is even, then the $q$-ary self-dual $2$-quasi negacyclic codes
of length~$2n$ always exist.
\end{theorem}

And for any $q$ we obtain that
the $q$-ary self-dual $2$-quasi negacyclic codes are asymptotically good.
Precisely, we have the following theorem.

\begin{theorem} \label{int asymptotic}
Let $\delta\in(0,1-q^{-1})$ and $h_q(\delta)<\frac{1}{4}$.
Then there are $q$-ary self-dual $2$-quasi negacyclic codes
$C_1 ,C_2, \dots$ (hence the rate ${\rm R}(C_i)=\frac{1}{2}$) such that
the code length $2n_i$ of $C_i$ goes to infinity
and the relative minimum weight $\Delta(C_i)>\delta$ for $i=1,2,\dots$.
\end{theorem}

The methodological innovation of the article are as follows.
Firstly, instead of the reciprocal polynomials, we use
an operator ``$*$'', which is an algebra automorphism of order $2$
on the quotient algebra $F[X]\big/\langle X^n+1\rangle$
(see Lemma~\ref{lem order of *} below).
Secondly, by analysing the $q$-cosets, for any odd $q$
we get the decomposition of
the quotient algebra $F[X]\big/\langle X^{2^m}+1\rangle$ with $m\geq 1$
(see Remark~\ref{rk r hat t ...} below).
Then, with the help of an estimation formula
on the so-called {\em balanced codes} (see Eq.\eqref{eq balanced <=} below),
we can estimate the number of the self-dual $2$-quasi negacyclic codes
whose relative minimum weights are small (see Lemma \ref{lem lem D <= delta}).

In Section~\ref {preliminaries}, we sketch some
algebraic and number-theoretic preparations.

In Section~\ref{operator on CR},
we introduce the operator ``$*$''
on the algebra $F[X]\big/\langle X^n+1\rangle$.

In Section~\ref {existence}, after some preparations,
we study the existence of self-dual $2$-quasi negacyclic codes over $F$.
The proof of the above Theorem~\ref{int existence}
divides into two  cases,
see Theorem~\ref{th existence odd} and Theorem~\ref{thm existence even}.

Section~\ref {asymptotic} is devoted to a proof of the above
Theorem~\ref{int asymptotic},
 it is relabelled as Theorem~\ref{thm asymptotic}.

Finally, conclusion is made in Section~\ref{Conclusions}.

\section{Preliminaries}\label{preliminaries}

Let $F$ be always the finite field of cardinality $|F|=q$,
where $q$ is a prime power.
In this article, any ring $R$ has identity $1_R$ (or $1$ for short), and
``subring'' and ``ring homomorphism'' are identity preserving.
By $R^\times$ we denote the multiplication group consisting of
all units (invertible elements) of the ring $R$. In particular,
$F^\times=F\setminus\{0\}$.
If a ring $R$ is also an $F$-vector space, then $R$ is called an $F$-algebra.
An algebra homomorphism between two algebras is both a ring homomorphism
and a linear map.
For any $F$-algebra $R$, there is a canonical injective homomorphism
$F \to R$, $\alpha\mapsto \alpha 1_R$, to embed $F$ into $R$,
so that we can write $F\subseteq R$ and say that $F$ is a subalgebra of $R$
(note that with the embedding $1_F=1_R$).

Let $n>1$ be an integer. A linear code $C$ of $F^n$ is called a {negacyclic code} if
\begin{align}\label{eq def negacyclic}
  (c_0,c_1,\dots,c_{n-1})\in C~ \implies~
  (-c_{n-1} ,c_0,\dots,c_{n-2})\in C;
\end{align}
and a linear code $C$ of $F^n\times F^n$
is called a {\em $2$-quasi negacyclic code}
if
\begin{align}\label{eq def 2-quasi negacylic}
\begin{array}{ccc}
& (c_0,c_1,\dots,c_{n-1},\; c_0',c_1',\dots,c_{n-1}')\in C \\
&\implies  (-c_{n-1},c_0,\dots,c_{n-2}, \; -c_{n-1}',c_0',\dots,c_{n-2}')\in C.
\end{array}
\end{align}

In the following we always denote
$\CR=F[X]/\langle X^n+1\rangle$ with $n>1$,
which is the quotient algebra of the polynomial algebra $F[X]$ over the ideal
$\langle X^n+1\rangle$ generated by $X^n+1$.
Any residue class modulo $X^n+1$ has a unique representative
polynomial with degree less than $n$. Hence we can write 
\begin{align} \label{eq R_lambda}
\CR=F[X]\big/\langle X^n +1\rangle
 =\big\{\,  a_0+a_1X+\dots+a_{n-1}X^{n-1}\;\big|\; a_i\in F\,\big\}.
\end{align}
Further, the Cartesian product
\begin{align}
 \CR^2=\CR\times\CR
 =\big\{\big(a(X),a'(X)\big)\;\big|\; a(X),~ a'(X)\in\CR \big\}
\end{align}
is an $\CR$-module.

\begin{remark}\label{rk identification} \rm
{\rm (1)}~There is a canonical linear ismorphism
\begin{align} \label{eq CR to F}
 {\CR}~
\mathop{\longrightarrow}^{\cong}~ F^n,\qquad
a(X)=\sum\limits_{i=0}^{n-1}a_{i}X^{i} ~\longmapsto~ \mathbf{a}=(a_0,a_1,\dots,a_{n-1}).
\end{align}
So any element $a(X)=a_0+a_1X+\dots+a_{n-1}X^{n-1}\in\CR$
is identified with a word
${\bf a}=(a_0,a_1,\dots,a_{n-1})\in F^n$.
Note that, for $a(X)=a_0+a_1X+\dots+a_{n-1}X^{n-1}\in\CR$, the element
$$
 Xa(X)=-a_{n-1}+a_0X+\dots+a_{n-2}X^{n-1}~({\rm mod}~X^n+1)
$$
corresponds to the word $(-a_{n-1},a_0,a_1,\dots,a_{n-2})\in F^n$.
By Eq.(\ref{eq def negacyclic}),
the ideals of $\CR$ ($\CR$-submodules of the regular module $\CR$)
 are identified with the negacyclic codes over $F$ of length $n$.

{\rm (2)}~Extending the linear isomorphism in Eq.(\ref{eq CR to F}) to
 the following linear isomorphism:
\begin{align} \label{eq CR CR to F F}
\CR\!\times\!\CR \mathop{\longrightarrow}^{\cong}~  F^n\!\times\! F^n, ~~
(a(X),a'(X))\mapsto (a_0,a_1,\dots, a_{n-1},\,a'_0,a'_1,\dots, a'_{n-1}),
\end{align}
where $a(X)=\sum\limits_{i=0}^{n-1}a_{i}X^{i}$
and $a'(X)=\sum\limits_{i=0}^{n-1}a'_{i}X^{i}$ in $\CR$.
With the isomorphism Eq.(\ref{eq CR CR to F F}),
it is easy to see that $(Xa(X),\,Xa'(X))$ corresponds to the word
$(-a_{n-1},a_0,a_1, \dots,a_{n-2},\,-a'_{n-1},a'_0, a'_1,\dots,a'_{n-2})\in F^n\times F^n$.
By Eq.(\ref{eq def 2-quasi negacylic}),
the $\CR$-submodules of $\CR^2=\CR\times\CR$
are identified with the $2$-quasi negacyclic codes over $F$ of length $2n$.
\end{remark}

Next we introduce more algebraic properties of
$\CR$ in the semisimple case,
i.e., $\gcd(2n,q)=1$; in this case $q$ is odd.
Let $\Z_{2n}$ denote the residual integer ring modulo $2n$;
so $q\in\Z_{2n}^\times$, where $\Z_{2n}^\times$ is the multiplication group
consisting of all the units (invertible element) of $\Z_{2n}$.
By ${\rm ord}_{\Z_{2n}^\times}(q)$ we denote the order of $q$
in the multiplication group $\Z_{2n}^\times$, and
by $\langle q\rangle_{\Z_{2n}^\times}$ we denote the
subgroup of the multiplication group $\Z_{2n}^\times$ generated by~$q$.
The group $\langle q\rangle_{\Z_{2n}^\times}$ acts on the set $\Z_{2n}$ by multiplication,
any orbit of the action is called a {\em $q$-orbit}, or a {\em $q$-coset}.
Let
$$
 1+2\Z_{n}=\{1+2k\,|\,k\in\Z_n\}=\{1,3,5,\dots, 2n-1\}\,\subseteq\,\Z_{2n}.
$$
Since $\gcd(2n,q)=1$ and $2\,|(q-1)$,
it follows that $q\in{\Z_{2n}^\times\cap (1+2\Z_{n})}$.
Thus the group $\langle q\rangle_{\Z_{2n}^\times}$ acts on the set $1+2\Z_n$,
 i.e.,
$q (1+2k)\in 1+2\Z_n$ for any $1+2k\in 1+2\Z_n$.
Denote the set of all $q$-orbits within $1+2\Z_{n}$ as follows
\begin{align} \label{eq q-orbits}
 (1+2\Z_n)\big/\langle q\rangle_{\Z_{2n}^\times} =\{Q_1,\;\dots,\; Q_t\}.
\end{align}
Set ${\cal T}=\{1,\dots,t\}$ and ${\cal T}_{\{i\}}= {\cal T}\setminus\{i\}$.
The decomposition of $\CR=F[X]\big/\langle X^n+1\rangle$ is related with the
 $q$-cosets in Eq.\eqref{eq q-orbits}.
The (1)---(2) of the following remark can be found
 in \cite{CDFL}.

\begin{remark} \label{rk q-coset} \rm
Assume that $\gcd(2n,q)=1$.
For an extension of $F$, there is a primitive $2n$-th root of unity
 $\xi$ such that $\xi^n=-1$.
It is easy to see that $\xi^i$, $i\in1+2\Z_{n}$, are all the roots of
the polynomial $X^n+1$, i.e.,
$X^n+1=\prod_{i\in(1+2\Z_n)}(X-\xi^i)$ (over $F(\xi)$, an extension of $F$).
Then:

{\bf(1)}~ For each $q$-orbit $Q_i$ in Eq.\eqref{eq q-orbits},
the polynomial $\psi_i(X)=\prod_{j\in Q_i}(X-\xi^j)$ is irreducible over~$F$.
And the following is an irreducible decomposition of $X^n+1$ over $F$:
\begin{align}
X^n+1=\psi_1(X)\cdot\dots\cdot \psi_t(X)=\prod_{i\in{{\cal T}}}\psi_i(X).
\end{align}
For $j_0\in Q_i$, we have an isomorphism:
$$
  F[X]\big/\langle \psi_i(X)\rangle
  \mathop{\longrightarrow}\limits^{\cong}F(\xi^{j_0}), ~~
   f(X)\longmapsto f(\xi^{j_0}),
$$
where $F(\xi^{j_0})$ is a field extension over~$F$ generated by $\xi^{j_0}$
with extension degree $|F(\xi^{j_0}):F|=\deg(\psi_i(X))$, i.e.,
$F[X]\big/\langle \psi_i(X)\rangle$ is the field with cardinality $q^{\deg(\psi_i(X))}$.

{\bf(2)}~ Let $\hat\psi_i(X)=(X^n+1)\big/\psi_i(X)
  =\prod_{k\in{{\cal T}_{\{i\}}}}\psi_k(X)$,
$i=1,\dots,t$. 
Then each $\CR\hat\psi_i(X)=\CR e_i$ is a minimal ideal of
$\CR$, where $e_i$ denotes the primitive idempotent of the minimal ideal.
Note that $e_1, e_2, \dots, e_t$ are all primitive idempotents of $\CR$.
And the following hold:
\begin{align} \label{eq R= Re_1+...}
\begin{array}{l}
 \CR=\CR\hat\psi_1(X)\oplus\dots\oplus \CR\hat\psi_t(X)=
\CR e_1\oplus\dots\oplus \CR e_t ;\\
 1_{\CR}=e_1+\dots+e_t, ~~~~
  e_i e_j=\begin{cases} e_i, & i=j; \\ 0, & i\neq j; \end{cases} \\
\CR e_i=\CR\hat\psi_i(X) \cong F[X]\big/\langle\psi_i(X) \rangle, \\
 \dim_F\CR e_i =\deg \psi_i(X) =|Q_i|,   ~~~ i=1,\dots,t,
\end{array}
\end{align}
where $\CR e_i$ is a extension field over $F$
with degree $|\CR e_i: F|=\dim_F\CR e_i$.

{\bf(3)}~Assume that $n=2^\ell$ for an integer $\ell \geq 1$.
It is obvious that
$$
 1+2\Z_{2^{\ell}}=\{1,\,3,\,\dots,\, 2^{\ell+1}-1 \}
 =\Z_{2^{\ell+1}}^\times
$$
is exactly the multiplication group of the units of $\Z_{2^{\ell+1}}$,
the order $|\Z_{2^{\ell+1}}^\times|=2^\ell$.
Note that $\langle q\rangle_{\Z_{2^{\ell+1}}^\times}$ is a subgroup of
$\Z_{2^{\ell+1}}^\times$, and the $q$-orbits are
just the cosets of the subgroup
$\langle q\rangle_{\Z_{2^{\ell+1}}^\times}$ in $\Z_{2^{\ell+1}}^\times$.
In this case, for $ i=1,\dots,t$
\begin{align} \label{eq |Q_i|=}
 |Q_i|={\rm ord}_{\Z_{2^{\ell+1}}^\times}(q), ~~{\rm and}~~
 t=|\Z_{2^{\ell+1}}^\times|\big/|Q_i|,
\end{align}
where $t$ is the number of the $q$-orbits as in Eq.\eqref{eq q-orbits} (with $n=2^\ell$).
\end{remark}
Note that $Z_{2^{2}}^\times=\{1,\,3\}=\{\pm 1\}$,
it is easy to check that ${\rm ord}_{\Z_{2^{2}}^\times}(q)=1$ or $2$.
Next, we consider ${\rm ord}_{\Z_{2^{\ell}}^\times}(q)$ in the case that $\ell>2$.
The group structure of $\Z_{2^{\ell}}^\times$
 is well-known, e.g., see~\cite[Theorem 13.19]{H}. We state it as a lemma.

\begin{lemma}\label{lem -1 . 5}
Let $\ell >2$ be an integer. Then
the multiplication group
$\Z_{2^{\ell}}^\times$ is a direct product of two cyclic subgroups:
\begin{align} \label{eq -1 . 5}
 \Z_{2^{\ell}}^\times=\langle -1\rangle_{\Z_{2^{\ell}}^\times}
  \times \langle 5\rangle_{\Z_{2^{\ell}}^\times},
\end{align}
and ${\rm ord}_{\Z_{2^{\ell}}^\times}(5)=2^{\ell -2}$.
\qed
\end{lemma}


\begin{lemma}\label{lem ord q}
Assume that $q$ is odd. Then there is an integer $\mu_q\ge 2$
such that for any integer $\ell> \mu_q$ the order
${\rm ord}_{\Z_{2^{\ell}}^\times}(q)=2^{\ell-\mu_q}$.
\end{lemma}
\begin{proof}
It is clear that $q\,{\equiv} \pm1\pmod{2^2}$.
If $q\,{\equiv} \pm1\pmod{2^k}$ for a positive integer~$k$,
then $q\,{\equiv} \pm1\pmod{2^{\ell}}$ for any integer $\ell$ with $1<\ell\le k$.
Obviously, for large enough integer $\ell$ we have
$1< q<2^{\ell}-1$, hence $q\,{\not\equiv} \pm1\pmod{2^{\ell}}$.
Thus there is an integer $\mu_q\ge 2$
such that
\begin{align} \label{eq m<mu_q}
\begin{array}{ll}
 q\,{\not\equiv} \pm 1 \!\pmod{2^{\ell}}, ~~~ & \forall\, \ell > \mu_q;\\[3pt]
 q\,{\equiv} \pm 1 \!\pmod{2^{\ell}}, & \forall\, \ell \le \mu_q.
\end{array}
\end{align}
where ``$q\,{\not\equiv} \pm 1 \!\pmod{2^{\ell}}$'' means that
$q\notin\{ \pm 1\}$ in $\Z_{2^\ell}^\times$.

Assume that $\ell>\mu_q$. By Eq.\eqref{eq -1 . 5},
in $\Z_{2^\ell}^\times$ we can write
\begin{align}\label{eq q=-1.5}
 q=(-1)^j\cdot 5^{i'2^{i}},
~~ \mbox{where $1\le i'2^i\le 2^{\ell-2}$ with $i'$ being odd,  $j=0,1$}.
\end{align}
It is easy to check that
${\rm ord}_{\Z_{2^{\ell}}^\times}(q) =
{\rm lcm}({\rm ord}_{\Z_{2^{\ell}}^\times}((-1)^j),\,
  {\rm ord}_{\Z_{2^{\ell}}^\times}(5^{i'2^i}))$,
where ``\rm lcm'' means the least common multiple.
Obverse that
$q\,{\not\equiv} \pm 1 \!\pmod{2^{\ell}}$
since $\ell >\mu_q$ see Eq.\eqref{eq m<mu_q}.
Thus
${\rm ord}_{\Z_{2^{\ell}}^\times}(q)={\rm ord}_{\Z_{2^{\ell}}^\times}(5^{i'2^{i}})$.
Denote
$$
 {\rm ord}_{\Z_{2^{\ell}}^\times}(q)=2^{r}, ~~~~ \mbox{where $0\le r\le \ell-2$}.
$$
So we can write $5^{i'2^{i}}=5^{i'2^{\ell-r-2}}$.
Note that $2^{\mu_q+1}\,|\, 2^{\ell}$ (since $\ell > \mu_q$), then
$$
q\equiv (-1)^j\cdot 5^{i'2^{\ell-r-2}}~({\rm mod}~{2^{\mu_q+1}}) ,
$$
By Eq.\eqref{eq m<mu_q}, we have
\begin{align*}
& (-1)^j\cdot 5^{i'2^{\ell-r-2}}\,{\not\equiv} \pm 1~({\rm mod}~{2^{\mu_q+1}}); \\
& (-1)^j\cdot 5^{i'2^{\ell-r-2}}\,{\equiv} \pm 1~({\rm mod}~{2^{\mu_q}}).
\end{align*}
Note that (cf. Lemma~\ref{lem -1 . 5})
$$
{\rm ord}_{\Z_{2^{\mu_q+1}}^\times}(5)=2^{\mu_q-1}, ~~~~~
{\rm ord}_{\Z_{2^{\mu_q}}^\times}(5)=2^{\mu_q-2}.
$$
It follows that
$$
\ell-r-2<\mu_q-1, ~~~~~
\ell-r-2\ge \mu_q-2,
$$
and so
$\ell-\mu_q -1  < r \le \ell-\mu_q$. That forces $r = \ell-\mu_q $.
\end{proof}

\begin{example} \rm
For some small $q$ we list the $\mu_q$ as follows.
$$\begin{array}{|c|ccccccccccccc|} \hline
q    & 3 & 5 & 7 & 9 & 11 & 13 & 17 & 19 & 23 & 25 & 27 & 29 & 31\\ \hline
\mu_q& 2 & 2 & 3 & 3 &  2  & 2  &  4  &  2 &  3 &  3  & 4  &  2  &  5\\ \hline
\end{array}$$
\end{example}

\begin{lemma} \label{lem -1 notin}
Assume that $q$ is odd. Let $\ell \geq 1$ be an integer. Then
$-1\in\langle q\rangle_{\Z_{2^{\ell}}^\times}$
if and only if 
$q=-1$ in $\Z_{2^{\ell}}^\times$;
and if it is the case, then $\langle q\rangle_{\Z_{2^{\ell}}^\times}=\{\pm 1\}$
and ${\rm ord}_{\Z_{2^{\ell}}^\times}(q)=2$.
\end{lemma}

\begin{proof}
The result holds obviously for $\ell=1$ or $2$.
When $\ell \geq 3$, we can
write $q=(-1)^j\cdot 5^{i' 2^{i}}$ as in Eq.\eqref{eq q=-1.5}.
Assume that $-1\in\langle q\rangle_{\Z_{2^{\ell}}^\times}$,
 there is an integer~$h$ such that $q^h=-1$ (in $\Z_{2^{\ell}}^\times$), i.e.,
$(-1)^{jh}\cdot 5^{i' 2^{i}h}=-1$.
Because $\Z_{2^{\ell}}^\times$ is a direct product of subgroups, see Eq.\eqref{eq -1 . 5}, 
we get that
$$
 (-1)^{jh}=-1~~~\mbox{and}~~~5^{i' 2^{i}h}=1
~~~~ \mbox{(in $\Z_{2^{\ell}}^\times$).}
$$
The former implies that $j=1$ and $h$ is odd; then
the latter implies that $2^i=2^{\ell-2}$. That is,
$q=-1$ in $\Z_{2^{\ell}}^\times$.
The converse is trivial.
\end{proof}

We sketch a result about the so-called {\em balanced codes}
 which we will quote.
Let $I=\{0,1,\dots, n-1\}$ be the index set, i.e., $F^n=F^I$.
For any subset $I_*=\{ i_1,\dots, i_k\}$ of $I$, i.e., $0\le i_1<\dots<i_k<n$,
we have the projection $\rho_{I_*}: F^I\to F^{I_*}$,
$(a_0,\dots,a_{n-1})\mapsto (a_{i_1},\dots,a_{i_k})$.
Let $C\subseteq F^I$. If there are positive integers $k,s,t$ and subsets
$I_1,\dots,I_s$ of $I$ with cardinality $|I_j|=k$ for $j=1,\dots,s$
satisfying the following two:

 --- for $j=1,\dots,s$, the restricted map $\rho_{I_j}|_{C}: C\to F^{I_j}$ is bijective;

 --- for any $i\in I$, the number of the index $j$ that $i\in I_j$ equals $t$;

\noindent then we say that $C$ is a {\em balanced code} with
information sets $I_1,\cdots,I_s$, and obviously, $|C|=q^k$.
Let $\delta\in (0,1-q^{-1})$ be a real number. For a balanced code
$C\subseteq F^I$ as above,
by \cite[Corollary 3.4]{FL15} we have that
(where $h_q(\delta)$ is the $q$-entropy function as in Eq.\eqref{eq def h_q}):
\begin{align} \label{eq balanced <=}
\textstyle
\big|C^{\le\delta}\big|\le q^{kh_q(\delta)}, ~~~~~
\mbox{where}~~
 C^{\le\delta} =\big\{ {\bf c}
 \,\big|\, {\bf c}\in C, \, \frac{{\rm w}({\bf c})}{n}\le \delta \big\}.
\end{align}
Negacyclic codes are balanced, see~\cite[Lemma~II.8]{FL22}.
From Eq.\eqref{eq balanced <=} and \cite[Corollary 3.5]{FL15}
(or cf. \cite[Lemma 2.7]{FL24}),
we have the following result.

\begin{lemma}\label{lem balance}
If $A$ is an ideal of $\CR$, then the $\CR$-submodule $A\times A$ of
$\CR^2=\CR\times\CR$ is a balanced code,
hence for $0<\delta< 1-q^{-1}$,
$$ 
 |(A\times A)^{\le\delta}| \le q^{\dim_F (A\times A)\cdot h_q(\delta)}.
\eqno\qed
$$ 
\end{lemma}

\section{The operator ``$*$'' on $\CR=F[X]\big/\langle X^n+1\rangle$}
 \label{operator on CR}
In this section, we introduce an operator ``$*$'' on $\CR$,
and investigate some properties of it in the semisimple case, i.e., $\gcd(2n,q)=1$.
Recall that $n>1$ and
\begin{align} \label{eq CR=...}
 \CR=F[X]\big/\langle X^n+1\rangle=\{a_0+a_1X+\dots+a_{n-1}X^{n-1}\,|\,a_i\in F\}.
\end{align}
It is known from \cite[Lemma 4.8]{FL24} that the following map
\begin{align} \label{eq * operator}
 *: ~~ \CR\,\to\, \CR,~~
 a(X) \; \mapsto \; a^*(X)=a(X^{2n-1})\!\!\pmod{X^n+1},
\end{align}
is an algebra automorphism
(it is a special case of \cite[Lemma 4.8]{FL24}).
We say that the algebra automorphism Eq.\eqref{eq * operator}
is an operator ``$*$'' on~$\CR$ .

Note that the symbol ``$X$'' stands specifically for the variable of polynomials.
In $\CR$ we denote ``$x$'' to be the element representative by the monomial $X$.
By~Eq.\eqref{eq CR=...},  $1, x, \dots, x^{n-1}$ are a basis of $\CR$ and
$x^n=-1$, i.e,.
\begin{align} \label{eq R=F+Fx+...}
\CR=
F\oplus Fx\oplus\dots\oplus Fx^{n-1}, ~~~~
\mbox{with relation}~~ x^n=-1.
\end{align}
Hence $x$ is an unit (invertible element) of $\CR$.
We can write $a(x)=\sum_{i=0}^{n-1}a_ix^i\in\CR$
and the expressions $x^{-1}$, $a(x^{-1})$ etc. make sense.
Since 
${\rm ord}_{\CR^\times}(x)=2n$ and
$x^{2n-1}=x^{-1}$,
the operator ``$*$'' in Eq.\eqref{eq * operator} can be rewritten as
\begin{align} \label{eq * operator x}
*:~\CR \,\to\, \CR, ~~~ a(x) \;\mapsto\;a^*(x)=a(x^{-1}).
\end{align}

\begin{lemma} \label{lem order of *}
The operator ``$*$'' in Eq.\eqref{eq * operator}
(in Eq.\eqref{eq * operator x} equivalently)
is an algebra automorphism of $\CR$ of order~$2$.
\end{lemma}
\begin{proof}
By \cite[Lemma 4.8]{FL24}, the operator ``$*$'' is
an algebra automorphism of $\CR$.
For any $a(x)\in\CR$, by Eq.\eqref{eq * operator x} we have
\begin{align*}
\Big(\sum_{i=0}^{n-1} a_ix^i\Big)^{**}
=\Big(\sum_{i=0}^{n-1} a_i (x^{-1})^i\Big)^*
=\sum_{i=0}^{n-1} (a_i ((x^{-1})^{-1})^i
=\sum_{i=0}^{n-1} a_ix^i.
\end{align*}
Thus the order of the operator ``$*$'' equals $1$ or $2$.
Since ${\rm ord}_{\CR^\times}(x)=2n$ and $n>1$,
 $x^{-1}\ne x$. So
the operator ``$*$'' is not the identity operator.
In conclusion, the order of the operator ``$*$'' equals $2$.
\end{proof}

In the semisimple case, the $e_1,\dots,e_t$ in Eq.\eqref{eq R= Re_1+...}
are all primitive idempotents of $\CR$;
correspondingly, $\CR e_i=\CR\hat \psi_i(X)$ for $i=1,\dots,t$ are
all minimal ideals of~$\CR$. 
Thus the operator ``$*$'' permutes the primitive idempotents.
We will determine this permutation on $e_1,\dots,e_t$ by analysing
the $q$-orbits.
Recall that in the semisimple case, 
the set of all $q$-cosets within $1+2\Z_n$ as follows
 $$
 (1+2\Z_n)\big/\langle q\rangle_{\Z_{2n}^\times} =\{Q_1,\;\dots,\; Q_t\},
 ~~~~~ (\mbox{see Eq.\eqref{eq q-orbits}}).
$$

\begin{lemma} \label{lem i to i^*}
Let notation be as in Eq.\eqref{eq R= Re_1+...}. Assume that $\gcd(2n,q)\!=\!1$.
Let $Q_i\in (1\!+\!2\Z_{n})\big/\langle q\rangle_{\Z_{2n}^\times}$ for $i=1,\dots, t$
as in~Eq.\eqref{eq q-orbits}.

{\bf(1)}
 $-Q_i=\{-v\,|\,v\in Q_i\}$ is still an $q$-orbit of $1+2\Z_n$; hence
  there is an index $i^*\in{{\cal T}=\{1,\dots,t\}}$ such that $-Q_i=Q_{i^*}$.
In particular, $i\mapsto i^*$ is a permutation on ${\cal T}=\{1,\dots,t\}$.

{\bf(2)} Let $-Q_i=Q_{i^*}$ as above. Then
$(\CR \hat\psi_i(x))^*=\CR\hat\psi_i^*( x)=\CR\hat\psi_{i^*}(x)$.

\end{lemma}

\begin{proof}
{\bf(1)}. Let $Q_i=\{v_0,qv_0,\cdots,q^{k-1} v_0\}$ (may be repetitive),
where $k={\rm ord}_{\Z_{2n}^\times}(q)$.
Then $-Q_i=\{-v_0, q(-v_0),\cdots, q^{k-1}(-v_0)\}$ is still a $q$-orbit.

{\bf(2)}.
Let $\tilde F=F(\xi)$ where $\xi$ is a primitive $2n$-th root of unity
(see Remark~\ref{rk q-coset}) and
$\tilde{\CR}:=\tilde F[X]\big/\langle X^n+1\rangle$.
By Eq.\eqref{eq * operator x},
in $\tilde{\CR}$ 
we have the following computation:
\begin{align*}
\psi_i^*(x)
&=\prod_{j\in Q_i}(x^{-1}-\xi^j)
 =\prod_{j\in Q_i}\xi^{j} x^{-1} (\xi^{-j}-x)\\
 &=(-1)^{d_i} x^{-d_i}\xi^{\sum_{j\in{ Q_i}}j}\prod_{j\in Q_i} (x -\xi^{-j}),
\end{align*}
where $d_i=|Q_i|=\deg\psi_i(X)$, cf. Eq.\eqref{eq R= Re_1+...}.
Denote $\xi^{\sum_{j\in{ Q_i}}j}=\gamma_i$, which is a nonzero element of $F$.
We get
\begin{align*}
\psi_i^*(x) =(-1)^{d_i} x^{-d_i}\gamma_i\prod_{j\in -Q_i} (x -\xi^{j})
  =(- x^{-1})^{d_i}\gamma_i\cdot\psi_{i^*}(x),
\end{align*}
where
the last equality holds since  $-Q_i=Q_{i^*}$.
Thus
$$
\hat\psi_i^*(x)=\Big(\prod_{k\in{{\cal T}_{\{i\}}}}\psi_k(x)\Big)^*
=\prod_{k\in{{\cal T}_{\{i\}}}}\psi_k^*(x)
=\prod_{k\in{{\cal T}_{\{i\}}}}( -x^{-1})^{d_k}
 \gamma_k\cdot \psi_{k^*}(x).
$$
Obviously, $\sum_{k\in{{\cal T}_{\{i\}}}}d_k=n-d_i$.
Set $\prod_{k\in{\cal T}_{\{i\}}}\gamma_k=\hat\gamma_i$.
Note that $k\in{{\cal T}_{\{i\}}}$
is equivalent to that $k^{*}\in{{\cal T}_{\{i^*\}}}$.
So
$$
\hat\psi_i^*(x)=(-x^{-1})^{n-d_i}\hat\gamma_i
\cdot\prod_{k\in{{\cal T}_{\{i\}}}}\psi_{k^*}(x)
=(-x^{-1})^{n-d_i}\hat\gamma_i\cdot\hat\psi_{i^*}(x).
$$
Note that $(-x^{-1})^{n-d_i}\hat\gamma_i$ is an invertible element of $\CR$.
We obtain that $\CR\hat\psi_i^*(x)=\CR\hat\psi_{i^*}(x)$.
%
\end{proof}

\begin{remark} \label{rk * permutation} \rm
 By Lemma \ref{lem i to i^*}(2),
$\CR e_i^*=(\CR e_i)^*=(\CR\hat\psi_i(x))^*=\CR\hat\psi_{i^*}(x)=\CR e_{i^*}$,
i.e.,
$$
 e_i^*=e_{i^*}, ~~~~~~~~ i=1,\dots,t.
$$
In a word, the operator ``$*$'' permutes the primitive idempotents $e_1,\cdots,e_t$
in Eq.\eqref{eq R= Re_1+...}
as the same as the multiplication by ``$-1$'' permutes the $q$-orbits
${(1+2\Z_{n})\big/\langle q\rangle_{\Z_{2n}^\times}}
  =\{Q_1,\dots,Q_t\}$ in Eq.\eqref{eq q-orbits}.
And by Lemma \ref{lem order of *},
this permutation on $Q_1,\;\dots,\; Q_t$ is of order $2$ or $1$.

\end{remark}

\begin{lemma} \label{lem -1 in q}
Keep the notation in Lemma~\ref{lem i to i^*}.

{\bf (1)} If $-1\in\langle q\rangle_{\Z_{2n}^\times}$, then
$Q_{i^*}=Q_i$ for $i=1,2,\dots,t$.

{\bf (2)} If $n= 2^\ell$ {\rm($2n=2^{\ell+1}$)} with an integer $\ell \geq 1$ and
$-1\notin \langle q\rangle_{\Z_{2^{\ell+1}}^\times}$,
then $Q_{i^*}\neq Q_i$ for $ i=1,2,\dots,t$.
\end{lemma}
\begin{proof}
{\bf (1)} Let $Q_i=\{v_0,qv_0,\cdots,q^{k-1} v_0\}$
where $k={\rm ord}_{\Z_{2n}^\times}(q)$.
Note that $\langle q\rangle_{\Z_{2n}^\times}=\{1,q,\dots,q^{k-1}\}$
is a finite group.
If $-1$ is an element of $\langle q\rangle_{\Z_{2n}^\times}$,
then $-1,-q, \dots, -q^{k-1}$ are just a reorder of $1,q,\dots,q^{k-1}$.
Thus
$$Q_{i^*}=-Q_i=\{-v_0, -qv_0,\cdots, -q^{k-1}v_0\}=Q_i.$$

{\bf (2)}
Note that $\langle q\rangle_{\Z_{2^{\ell+1}}^\times}$
is a subgroup of $\Z_{2^{\ell+1}}^\times$, see Remark~\ref{rk q-coset}(3).
Denote $G=\Z_{2^{\ell+1}}^\times$,
$H=\langle q\rangle_{\Z_{2^{\ell+1}}^\times}$ and $b=-1$ for convenience.
Then $H$ is a subgroup of $G$,
and all the cosets of  $H$ in $G$ are just all the $q$-orbits.
For any coset $Ha$ ($a\in G$), $Ha=bHa=Hba$ if and only if
$b=(ba)a^{-1}\in H$. Thus, if $b\notin H$ then $b(Ha)\ne Ha$ for any coset $Ha$.
\end{proof}

About Lemma~\ref{lem -1 in q}(2), the condition ``$n=2^\ell$'' is necessary.
We now present a counterexample.
\begin{example} \rm
Take $q=3$, $n=10$. Then
$\langle q\rangle_{\Z_{2n}^\times}=\{1,3,9,7\}$, so
$-1\notin \langle q\rangle_{\Z_{2n}^\times}$.
The $\{5,15\}$ is a $q$-orbit of $1+2\Z_{10}\subseteq\Z_{20}$.
But $-\{5,15\}=\{-5,-15\}=\{5,15\}$.
\end{example}

Therefore, the case ``$n=2^m$, $m\ge 1$, and $q$ is odd'' is easy to control.
Assume that it is the case.
Instead of $\CR$, we denote specifically
\begin{align*}
 {\cal A}=F[X]\big/\langle X^{2^m}+1\rangle.
\end{align*}
Set ${\rm ord}_{\Z_{2^{m+1}}^{\times}}(q)=2^r$ ($0\le r\le m-1$).
By Remark~\ref{rk q-coset}(3),
 $1+2\Z_{2^m}=\Z_{2^{m+1}}^\times=\{1,3,\dots, 2^{m+1}-1\}$
with cardinality $|\Z_{2^{m+1}}^\times|=2^m$,
and the $q$-orbits in the following set 
\begin{align} \label{eq Z_2^m/q}
 \Z_{2^{m+1}}^\times\big/\langle q\rangle_{\Z_{2^{m+1}}^\times}
=\{Q_1,\dots,Q_t\}
\end{align}
are just the cosets of the subgroup
$\langle q\rangle_{\Z_{2^{m+1}}^\times}$
in the multiplication group $\Z_{2^{m+1}}^\times$;
and the cardinality of the cosets
\begin{align} \label{eq Q r=...}
 |Q_i|={\rm ord}_{\Z_{2^{m+1}}^{\times}}(q)=2^r,
 ~~~~~~ i=1,\dots,t;
\end{align}
hence
\begin{align}
t=|\Z_{2^{m+1}}^\times|\big/{\rm ord}_{\Z_{2^{m+1}}^{\times}}(q)
=2^{m-r}.
\end{align}
%
By Eq.(\ref{eq R= Re_1+...}), we write
\begin{align}\label{eq A= Ae 1+...}
   \CA =F[X]\big/\langle X^{2^m}+1\rangle
 =\CA e_1 \oplus \CA e_2 \oplus\dots \oplus \CA e_t,
 \end{align}
 where each $\CA e_i=\CA \hat\psi_i(X)$ corresponds the $q$-coset $Q_{i}$.
With the notation listed above, we summarize
the action of the operator ``$*$'' on the primitive idempotents as follows.

\begin{remark} \label{rk r hat t ...} \rm
Assume that $n=2^m$, $m\ge 1$, and $q$ is odd.

{\bf(1)} If $q\equiv -1~({\rm mod}~2^{m+1})$, then
by Lemma~\ref{lem -1 in q}(1) and Remark \ref{rk * permutation}, we have
 $e_i^*=e_{i}$ for $i=1,\dots,t$,
and so $(\CA e_i)^{*}=\CA e_i^*=\CA e_i$; in this case
we say that  $\CA e_i$ is a {\em ``$*$''-stable}
algebra (not a subalgebra of $\CA$ in general).
The ideal $\CA e_i$ is a field extension over $F$.

{\bf(2)} If $q\,{\not\equiv} -\!1~({\rm mod}~2^{m+1})$
(i.e.,  $-1\notin \langle q\rangle_{\Z_{2^{m+1}}^\times}$,
see Lemma~\ref{lem -1 notin}), then
by Lemma~\ref{lem -1 in q}(2) and Remark \ref{rk * permutation},
the operator ``$*$'' partitions all the primitive idempotents of $\CA$ into
$t/2=2^{m-r-1}$ pairs (for convenience we denote $\hat t=t/2$):
\begin{align} \label{eq hat t=...}
 e_1,\; e_1^*,\; \dots,\; e_{\hat t},\; e_{\hat t}^*,
~~~~~~~ \hat t=2^{m-r-1}.
\end{align}
Thus ${\cal A}=F[X]\big/\langle X^{2^m}+1\rangle$ can be rewritten as
\begin{align}\label{eq CA ei}
\begin{array}{ll}
& \CA=\CA e_1\oplus \CA e_1^* \oplus\; \cdots\;\oplus
 \CA e_{\hat t}\oplus \CA e_{\hat t}^*\\[5pt]
& \dim_F \CA e_i=\dim_F \CA e_i^* =|Q_i|=2^{r}, ~~~~ i=1, \dots, \hat t.
\end{array}
\end{align}
The $\CA e_i,\, \CA e_i^*$, $i=1,\dots, \hat t$, are all minimal ideals of $\CA$.
Further, set $\widehat e_i= e_i +  e_i^*$, $i=1,\dots, \hat t$.
Then $\widehat e_i^{*}= e_i^* +  e_i^{**}=\widehat e_i$,
$\CA\widehat e_i=\CA e_i\oplus\CA e_i^*$, and
\begin{align} \label{eq A=A widehat e}
\begin{array}{ll}
\CA=\CA \widehat e_1\oplus \; \cdots \; \oplus \CA \widehat e_{\hat t}, ~
&
\dim_F \CA\widehat e_i= 2^{r+1}, ~~ i=1,\dots, \hat t. \\[3pt]
1=\widehat e_1+\dots+\widehat e_{\hat t},
&
\widehat e_i \widehat e_j
 =\begin{cases} \widehat e_i, & i=j;\\ 0, & i\neq j. \end{cases}
\end{array}
\end{align}
The $\CA\widehat e_i$, $i=1,\dots, \hat t$,
are all minimal ``$*$''-stable ideals of $\CA$, i.e.,
 $(\CA\widehat e_i)^{*}=\CA\widehat e_i$.
By Eq.\eqref{eq A=A widehat e}, any $a\in\CA$ can be written as
$$
 a=a_1+\dots+a_{\hat t}, ~~~~\mbox{where }~ a_i=a\widehat e_i, ~~ i=1,\dots, \hat t.
$$
With this expression, for any $a,b\in\CA$ we have
\begin{align} \label{eq ab=a_1b_1+...}
  ab=a_1b_1+\dots+a_{\hat t}b_{\hat t}.
\end{align}

{\bf(3)} If $m\geq \mu_q$ where $\mu_q$ is defined in Eq.\eqref{eq m<mu_q},
then 
$q\,{\not\equiv}-\!1~({\rm mod}~2^{m+1})$ (since $m+1>\mu_q$),
so that all conclusions in the above (2) are applied.
And by Lemma~\ref{lem ord q}, we have that
\begin{align} \label{eq r=...}
r=m-\mu_q+1 ~~(\mbox{cf. Eq.\eqref{eq Q r=...}}),
~~~ \mbox{hence } ~\hat t=2^{\mu_q-2}
~(\mbox{cf. Eq.\eqref{eq hat t=...}}).
\end{align}
\end{remark}

\section{Existence of self-dual $2$-quasi negacyclic codes} \label{existence}
In this section $\CR=F[X]\big/\langle X^n+1 \rangle$
as before, cf. Eq.\eqref{eq R_lambda}.
We consider the existence of the self-dual $2$-quasi negacyclic codes
of length $2n$ over the finite field $F$.
Our discussion divide into two cases: $n$ is odd or $n$ is even.

With the identification Eq.(\ref{eq CR to F}),
the inner product on $\CR$ is as follows:
for $a(X)=\sum\limits_{i=0}^{n-1}a_{i}X^{i}$, 
$b(X)=\sum\limits_{i=0}^{n-1}b_{i}X^{i}\in\CR$,
\begin{align*}
\big\langle a(X),b(X)\big\rangle=a_0b_0+\dots+a_{n-1}b_{n-1}.
\end{align*}
Similarly, with the identification Eq.(\ref{eq CR CR to F F}),
the inner product on $\CR ^{2}=\CR \times \CR$ is as follows:
for $(a(x),\,b(x)),\,(a'(x),\,b'(x))\in{\CR ^{2}}$,
\begin{align*}
\big\langle (a(X),b(X)),\, (a'(X),b'(X))\big\rangle=\sum_{i=0}^{n-1}a_ia'_i +\sum_{i=0}^{n-1}b_ib'_i,
\end{align*}
where $a(X)\!=\!\sum\limits_{i=0}^{n-1}a_{i}X^{i},
b(X)\!=\!\sum\limits_{i=0}^{n-1}b_{i}X^{i}, a'(X)\!=\!\sum\limits_{i=0}^{n-1}a'_{i}X^{i},
b'(X)\!=\!\sum\limits_{i=0}^{n-1}b'_{i}X^{i}$ in~$\CR$.
The dual codes and the self-dual codes are explained as in Introduction.
It is well-known that for any ($2$-quasi) negacyclic code $C$,
the dual code $C^\bot$ is still a ($2$-quasi) negacyclic code,
cf. \cite{D12, DL} (or cf. \cite[Lemma 3.10]{FL24}).

In the following, we denote
 $$\CR_1=F[X]\big/\langle X^n-1\rangle.$$
Similarly to Remark~\ref{rk identification},
the ideals of $F[X]\big/\langle X^n-1\rangle$ are identified with
the cyclic codes over $F$ of length $n$;
and the $\CR_1$-submodules of $\CR_1^2=\CR_1\times \CR_1$
are identified with the $2$-quasi cyclic codes over $F$ of length $2n$.

If $q$ is even, then $1=-1$ and the ($2$-quasi) negacyclic codes
over $F$ are just the ($2$-quasi) cyclic codes over $F$.
If $n$ is odd, then
 the ($2$-quasi) negacyclic codes over $F$ 
are also related perfectly to the ($2$-quasi) cyclic codes over $F$. 
Precisely, we state it in the following paragraph.

Assume that $n$ is odd,
it is known (see \cite{D12}, or \cite[Lemma 6.1, Corollary 6.2]{FL24}) that
the following map
\begin{align*}
\eta: ~ \CR_1 \, \to\, \CR, 
 ~~~ a(X)\; \mapsto\; a(-X)~({\rm mod}~{X^n+1}),
\end{align*}
is an algebra isomorphism and
satisfies that for any $a(X),a'(X)\in \CR_1$ the following two hold:
\begin{itemize}
\item[(i)]
 the weight ${\rm w}\big(\eta(a(X))\big)={\rm w}\big(a(X)\big)$;
\item[(ii)]
the inner product
 $\big\langle\eta(a(X)),\eta(a'(X))\big\rangle=\big\langle a(X),a'(X)\big\rangle$.
\end{itemize}
And the algebra isomorphism $\eta$ is extended to the following map
\begin{align}\label{eq eta 2}
\eta^{(2)}: ~ \CR_1^2 \, \to\, \CR^2, 
 ~~~ \big(a(X),a'(X)\big)\; \mapsto\; \big(\eta(a(X)),\eta(a'(X))\big),
\end{align}
which is a module isomorphism and
preserves the weight and the inner product similarly to the above (i) and (ii).
With the notation as above, we have the following
(cf. \cite[Theorem 6.3]{FL24}).

\begin{remark} \label{rk n is odd} \rm
 Assume that $n$ is odd.

{\bf(1)} A linear code $C$ of $\CR_1^2$ is a $2$-quasi cyclic codes of length $2n$
if and only if $\eta^{(2)}(C)$ is a $2$-quasi negacyclic code of length $2n$;
and in that case, we have that
${\rm R}\big(\eta^{(2)}(C)\big)={\rm R}\big(C\big)$ and
$\Delta\big(\eta^{(2)}(C)\big)=\Delta\big(C\big)$,
and $\eta^{(2)}(C)$ is self-dual if and only if $C$ is self-dual.

{\bf(2)}
Note that $\eta$ and $\eta^{(2)}$ are isomorphisms.
If $C$ of $\CR^2$ is a self-dual $2$-quasi negacyclic code of length $2n$,
then the inverse image of $C$ in $\CR_1^2$
is a self-dual $2$-quasi cyclic code of length $2n$.
\end{remark}

By Remark~\ref{rk n is odd}, we get that
$q$-ary self-dual $2$-quasi negacyclic codes of length~$2n$ ($n$ is odd) exist
if and only if self-dual $2$-quasi cyclic codes of length~$2n$ exist.
By \cite[Proposition 6.1]{LS01} (for semisimple case) and
\cite[Theorem 6.1]{LS03} (for general case),
the $q$-ary self-dual $2$-quasi cyclic codes of length $2n$ exist if and only if
$q\,{\not\equiv}-\!1~({\rm mod}~4)$.
Thus, we get the following theorem immediately.

\begin{theorem}\label{th existence odd}
 If $n$ is odd, then $q$-ary self-dual $2$-quasi negacyclic codes
of length $2n$ exist if and only if $q\,{\not\equiv}-\!1~({\rm mod}~4)$.
\end{theorem}

It is well-known that $\CR$ is a principal ideal ring,
i.e., any ideal of $\CR$ is generated by one element.
Note that an $\CR$-submodules of $\CR^2$
may be not generated by one element.
However, the $\CR$-submodules of $\CR^2$ generated by one element
will play an active role in this article.
For $(a,b)=(a(X),b(X))\in\CR^2$, by
\begin{align} \label{eq C_a,b}
 C_{a,b}=\big\{f(X)\big(a(X),b(X)\big)\,\big|\, f(X)\in\CR\big\}
\end{align}
we denote the $2$-quasi negacyclic code generated by $(a,b)$.
If $a$ is invertible, then it is easy to check that
$C_{a,b}=C_{1,g}$, where $g=a^{-1}b$ and $C_{1,g}$
is generated by $(1,g)=(1,g(X))\in\CR^2$.
There is an effective way to recognize if $C_{1,g}$ is self-dual.
With the notation in Eq.\eqref{eq C_a,b} and Eq.\eqref{eq * operator}
(or Eq.\eqref{eq * operator x}), we have:

\begin{lemma} \label{lem * and self-dual}
The $2$-quasi negacyclic code $C_{1,g}$ for $(1,g)\in\CR^2$
is self-dual if and only if $gg^*=-1$.
\end{lemma}

\begin{proof}
This is a special case of \cite[Lemma 4.9]{FL24}.
\end{proof}

\begin{remark} \label{rk D(A,*)} \rm
For the algebra $\CR$ with the operator ``$*$'' as in Eq.\eqref{eq * operator x},
we denote
\begin{align*}
\CD(\CR,*) = \big\{\, g\,\big|\, g\in \CR,\, gg^*=-1 \big\}.
\end{align*}
By Lemma~\ref{lem * and self-dual}, if $\CD(\CR,*)\neq\emptyset$
then for any  $g\in\CD(\CR,*)$ the $2$-quasi negacyclic code
$C_{1,g}$ over $F$ of length $2n$ is self-dual.
And the following hold.

{\bf(1)}
If $A$ is a {\em ``$*$''-stable} subalgebra of $\CR$
(hence $1_A=1_{\CR}$), i.e. $A^*=A$,
then the restriction of the operator ``$*$'' on $A$
is an algebra automorphism of $A$, and we denote
$ 
\CD(A,*) = \{ g\,|\, g\in A,\, gg^*=-1 \}.
$ 
Obviously, $\CD(A,*)\subseteq\CD(\CR,*)$
for any ``$*$''-stable subalgebra $A$ of $\CR$.

{\bf(2)}
If $\CR=\CR e\oplus\CR f$,
where $e^2=e$, $f^2=f$, $ef=0$ and $1=e+f$, then $\CR e$ is an algebra
with identity $e$ (but not a subalgebra of $\CR$ as $e\ne 1_{\CR}$).
It is the same for $\CR f$. Further, if $\CR e$ and $\CR f$ are
``$*$''-stable, then $\CD(\CR e,*)$ and $\CD(\CR f,*)$ make sense, e.g.,
$\CD(\CR e,*)= \{ g\,|\, g\in \CR e,\, gg^*=-e \}$; and
$$
 \CD(\CR, *)= \big\{\, g_e+g_f\,\big|\,
 g_e\in\CD(\CR e,*),\, g_f\in\CD(\CR f,*)\big \}.
$$
In particular, the cardinality
$|\CD(\CR, *)|=|\CD(\CR e, *)|\cdot|\CD(\CR f, *)|$.
\end{remark}

The following lemma shows a sufficient condition for the existence of self-dual
$2$-quasi negacyclic codes over $F$ of length $2n$.

\begin{lemma} \label{lem q D(F,*)}
If $q\,{\not\equiv}-\!1~({\rm mod}~4)$, then
$F=F\!\cdot\! 1_{\CR}$ is a ``$*$''-stable subalgebra of~$\CR$
and $|\CD(F,*)|=1$ or $2$, hence $|\CD(\CR,*)|\geq 1$.
\end{lemma}
\begin{proof}
The $F\!\cdot\! 1_\CR=\{\alpha\!\cdot\! 1_{\CR}
   \;|\; \alpha\in F\}$ is a subalgebra of $\CR$
   invariant by the automorphism $*$.
By Eq.\eqref{eq * operator},
the restriction of  ``$*$'' to $F\cdot 1_{\CR}$ is the identity automorphism of~$F$,
i.e., $g^*=g$ for any $g\in F\cdot 1_{\CR}$; hence $gg^*=g^2$.
It follows that
$|\CD(F,*)|$ equals the number of the solutions in~$F$ of the equation $X^2=-1$.
Note that $q\,{\not\equiv}-\!1~({\rm mod}~4)$ is equivalent to
that either $q$ is even or $4\,|\,(q-1)$.
If $q$ is even, the equation $X^2=-1$ has one solution in $F$.
If $4\,|\,(q-1)$, the equation $X^2=-1$ has two solutions in $F$.
Thus $|\CD(\CR,*)|\geq1$
since $\CD(F,*)\subseteq\CD(\CR,*)$.
\end{proof}

Next we calculate $|\CD({\cal A},*)|$,
where ${\cal A}=F[X]\big/\langle X^{2^m}+1\rangle$ with $m\geq 1$;
in this case, recall that
\begin{itemize}
\item
$(1+2\Z_{2^m})\big/\langle q\rangle_{\Z_{2^{m+1}}^\times}
=\Z_{2^{m+1}}^\times\big/\langle q\rangle_{\Z_{2^{m+1}}^\times}
=\{Q_1,\cdots,Q_t\}$, see Eq.\eqref{eq Z_2^m/q};
\item
$\CA={\cal A}e_1\oplus\dots\oplus{\cal A}e_t$,
see Eq.\eqref{eq A= Ae 1+...}.
\end{itemize}

\begin{lemma} \label{lem |D(A,*)|=}
Assume that $q$ is odd and
${\cal A}=F[X]\big/\langle X^{2^m}\!+1\rangle$
with $m\geq 1$. Then

{\bf(1)} If $q\equiv -1\pmod{2^{m+1}}$,
then $|\CD({\cal A},*)|=(q+1)^{2^{m-1}}$.

{\bf(2)} If $q\,{\not\equiv}-\!1 \pmod{2^{m+1}}$,
then $|\CD({\cal A},*)|=(q^{2^r}-1)^{\hat t}$,
where $2^r={\rm ord}_{\Z_{2^{m+1}}^\times}(q)$ as in Eq.\eqref{eq Q r=...}
and $\hat t=2^{m-r-1}$ as in Eq.\eqref{eq hat t=...}.
\end{lemma}
\begin{proof}
{\bf (1)}. By Lemma~\ref{lem -1 notin},
 $q\equiv -1\pmod{2^{m+1}}$ implies that
$\langle q\rangle_{\Z_{2^{m+1}}^\times}=\{\pm 1\}$.
Thus, for $i=1,\dots, t$, 
$Q_i=\{ j, -j\}=Q_{i^*}$ and
$\psi_i(X)=(X-\xi^j)(X-\xi^{-j})$ 
is irreducible over~$F$.
Note that $e_i^*=e_i$ and $\CA e_i$ is ``$*$''-stable, cf. Remark \ref{rk r hat t ...}(1).
%
By Eq.\eqref{eq R= Re_1+...},
$$
 {\cal A}e_i\cong F[X]\big/\langle \psi_i(X)\rangle
  \mathop{\longrightarrow}\limits^{\cong}F(\xi^j), ~~~~
   f(X)\longmapsto f(\xi^j),
$$
where $F(\xi^j)$ is the field extension generated by $\xi^j$ over~$F$.
So $\CA e_i\cong F_{q^2}$ is the field with cardinality $q^2$.
Because $x^*=x^{-1}$ (see Eq.\eqref{eq * operator x}),
the restriction of ``$*$'' to ${\cal A}e_i$ induces
the automorphism of $F(\xi^j)$ mapping $\xi^j$ to $\xi^{-j}$.
In a word, the restriction of $*$ to ${\cal A}e_i$ induces an automorphism
of ${\cal A}e_i$ of order~$2$. Thus
$$
 a^*=a^q, ~~~~~\forall\, a\in{\cal A}e_i.
$$
Then for $g_i\in\CA e_i$,
$$
 g_ig_i^*\!=-e_i ~\iff ~ g_i^{1+q}=-e_i ~ \iff ~
 \mbox{$g_i$ is a root of the polynomial $X^{q+1}+1$.}
$$
The multiplication group $F_{q^2}^\times$ is a cyclic group of order
$$
 |F_{q^2}^\times|=q^2-1=(q+1)(q-1).
$$
And $2\,|\,(q-1)$ (since $q$ is odd). So $F_{q^2}^\times$ has a
subgroup $L$ of order $2(q+1)$. The elements of $L$ are exactly
all roots of the polynomial $X^{2(q+1)}-1$. And
$$
X^{2(q+1)}-1=(X^{q+1}+1)(X^{q+1}-1)
$$
Thus all roots of the polynomial $X^{q+1}+1$ are inside $F_{q^2}$.
In conclusion,
$$
 |\CD({\cal A}e_i,*)|=q+1, ~~~~~~ i=1,\dots, t.
$$
Since $\dim_F{\cal A}e_i=2$, the number $t=2^m/2=2^{m-1}$.
Then by Remark~\ref{rk D(A,*)}(2),
$$
 |\CD({\cal A},*)|=\prod_{i=1}^t|\CD({\cal A}e_i,*)|
=(q+1)^t=(q+1)^{2^{m-1}}.
$$

{\bf (2)}.
By Lemma~\ref{lem -1 notin} again,
 $q~{\not\equiv}-\!1\pmod{2^{m+1}}$ implies that
$-1\notin\langle q\rangle_{\Z_{2^{m+1}}^\times}$;
in this case for $i=1,\dots, t$, the $Q_i\neq Q_{i^*}$
and $e_i^*\neq e_i$ (see Lemma~\ref{lem -1 in q}(2)).
From Eq.\eqref{eq A=A widehat e} we have that
\begin{align*}
{\cal A}={\cal A}\widehat e_1
 \oplus\dots\oplus {\cal A}\widehat e_{\hat t} ,
 ~~~~ \mbox{where} ~ {\cal A}\widehat e_i={\cal A}e_i\oplus{\cal A}e_i^* , ~
 i=1,\dots, \hat t.
\end{align*}
For $g_i=g_i'+g_i''\in\CA\widehat e_i$,
where $g_i'\in\CA e_i$ and $g_i''\in\CA e_i^*$,
we have $g_i^*=g_i''^*+g_i'^*$ with $g_i''^*\in\CA e_i$ and $g_i'^*\in\CA e_i^*$.
Then $g_ig_i^*=(g_i'+g_i'')(g_i''^*+g_i'^*)=g_i'g_i''^*+g_i'^*g_i''$.
So, $g_ig_i^*=-\widehat e_i=-e_i-e_i^*$ if and only if
$g_i'g_i''^*=-e_i$ and $g_i'^*g_i''=-e_i^*$.
We take any $g_i'\in(\CA e_i)^\times$, then $g_i''^*=-g_i'^{-1}$
(where $g_i'^{-1}$ is the inverse of $g_i'$ in $\CA e_i$, not in $\CA$),
hence $g_i''=g_i''^{**}=-(g_i'^*)^{-1}$
is uniquely determined. In a word,
\begin{align} \label{eq aa^*=...}
 g_ig_i^*=-\widehat e_i ~\iff~
 g_i = g_i' -(g_i'^*)^{-1}~~ \mbox{for a}~ g_i'\in(\CA e_i)^\times.
\end{align}
 Since $\dim_{F}\CA e_i=|Q_i|={\rm ord}_{\Z_{2^{m+1}}^\times}(q)=2^r$
(cf. Eq.\eqref{eq Q r=...}),
$\CA e_i$ is a field with cardinality $|\CA e_i|=q^{2^r}$. We get that
\begin{align} \label{eq D A widehat e_i}
 |\CD(\CA\widehat e_i,*)|=|(\CA e_i)^\times|=q^{2^r}-1.
\end{align}
By Remark~\ref{rk D(A,*)}(2),
$$
|\CD({\cal A},*)|=\prod_{i=1}^{\hat t}|\CD({\cal A}\widehat e_i,*)|
=\prod_{i=1}^{\hat t}|(\CA e_i)^\times|
=(q^{2^r}-1)^{\hat t}.
$$
We are done.
\end{proof}

We now prove the following theorem to complete
a proof of Theorem~\ref{int existence}.

\begin{theorem} \label{thm existence even}
If $n$ is even, then $q$-ary self-dual $2$-quasi negacyclic codes
of length $2n$ always exist.
\end{theorem}
\begin{proof}

Let $\CR=F[X]\big/\langle X^n+1\rangle$.
If $q$ is even, then $\CD(\CR,*)\ne\emptyset$ (see Lemma~\ref{lem q D(F,*)});
hence the self-dual $2$-quasi negacyclic codes exist, cf. Remark~\ref{rk D(A,*)}.
In the following we assume that $q$ is odd, and
let $x=X$ be the invertible element of $\CR$
as in  Eq.\eqref{eq R=F+Fx+...} and Eq.\eqref{eq * operator x}, i.e.,
$\CR=F\oplus Fx\oplus\dots\oplus Fx^{n-1}$ with relation $x^{n}=-1$,
and $a^*(x)=a(x^{-1})$ for all $a(x)\in\CR$.


Let $n=2^m n'$ with $n'$ being odd. Since $n$ is even, $m\ge 1$.
Set $y=x^{n'}$. The following $2^m$ elements
$$
 1, \; y=x^{n'},\; y^2=x^{2n'}, \; y^3=x^{3n'}, ~
  \dots, ~ y^{2^m-1}=x^{2^m n' -n'},
$$
form a subset of $\{1,x,x^2,\dots,x^{n-1}\}$, hence they are
linearly independent; and $y^{2^m}=x^{2^m n'}=-1$.
Then in $\CR$ the following $2^m$-dimensional subspace
$$
\CA=F\oplus Fy\oplus Fy^2\oplus\dots\dots\oplus F y^{2^m-1}
$$
is multiplication closed. So the above $\CA$ is an algebra generated by $y$
with the relation $y^{2^m}=-1$.
And, for $a(y)=\sum_{i=0}^{2^m-1}a_iy^i\in\CA$,
$$
 a^*(y)=a^*(x^{n'})=a\big((x^{-1})^{n'}\big)
 =a\big((x^{n'})^{-1}\big)=a(y^{-1})
 \in\CA.
$$
So we get the following.
\begin{itemize}
\item
${\cal A}$ is a ``$*$''-stable subalgebra of $\CR$
(note that $1_{\CA}=1_F=1_{\CR}$);
\item
${\cal A}\cong F[X]\big/\langle X^{2^m}+1\rangle$, ~ $m\ge 1$.
\end{itemize}
By Lemma~\ref{lem |D(A,*)|=}, $|\CD({\cal A},*)|>0$ for any odd $q$.
 Since $\CD({\cal A},*)\subseteq\CD(\CR,*)$, we have $|\CD(\CR,*)|>0$,
and so the $q$-ary self-dual $2$-quasi negacyclic codes of length $2n$ exist.
\end{proof}

\section{Self-dual $2$-quasi negacyclic codes are good} \label{asymptotic}

In this section we prove Theorem~\ref{int asymptotic}.
If $q$ is even, then ``negacyclic'' is just ``cyclic''; and
the asymptotically good sequence of self-dual $2$-quasi cyclic codes
$C_1,C_2,\dots$ stated as in Theorem \ref{int asymptotic}~ has been constructed
in \cite[Theorem~IV.17]{LF22}.

In the following, we always assume that
\begin{itemize}
\item
$q$ is odd and $n=2^m$, $m\ge \mu_q$,
where $\mu_q$ is defined in Eq.\eqref{eq m<mu_q},
and in this case $q~{\not\equiv}-\!1\pmod{2^{m+1}}$;
\item
$\delta$ is a real number satisfying
$\delta\in(0,1-q^{-1})$ and $h_q(\delta)<\frac{1}{4}$,
where $h_q(\delta)$ is the $q$-entropy function defined in Eq.\eqref{eq def h_q}.
\end{itemize}
With the above assumptions, the (2) and (3) of
 Remark~\ref{rk r hat t ...} are applied to~$\CA=F[X]\big/\langle X^{2^m}+1\rangle$;
in particular,
\begin{align}\label{eq hat t= r=...}
  {\rm ord}_{\Z_{2^{m+1}}^\times}(q)=2^r ~ \mbox{ with }~
  r=m-\mu_q+1, ~~~\mbox{ and }~ \hat t=2^{\mu_q-2},
\end{align}
see Eq.\eqref{eq r=...},
where $\hat t$ is the number of the minimal ``$*$''-stable ideals of $\CA$
as shown in Eq.\eqref{eq A=A widehat e}:
$$
 \CA=\CA \widehat e_1\oplus \; \cdots \; \oplus \CA \widehat e_{\hat t},
 ~~~~~~~ \widehat e_i=e_i+ e_i^*, ~~ i=1,\dots,\hat t.
$$
For $i=1,\dots,\hat t$, $\CA\widehat e_i$ is a ``$*$''-stable algebra with identity $\widehat e_i$
(it is not a subalgebra of $\CA$ in general because $\widehat e_i\ne 1_{\CA}$ in general)
and $\dim_F \CA\widehat e_i= 2^{r+1}$,
then $\CD(\CA\widehat e_i,*)=\{ g_i\in\CA\widehat e_i\,|\,g_ig_i^*=-\widehat e_i\}$
 makes sense.
 By Eq.\eqref{eq ab=a_1b_1+...}
and Remark~\ref{rk D(A,*)}, we have that
\begin{align} \label{eq D(A,*)=...}
\CD(\CA,*)=
 \big\{\,g_1+\dots+g_{\hat t}\;\big|\; g_i\in\CD(\CA\widehat e_i,*),\,
   i=1,\dots,\hat t\,\big\}.
\end{align}

For any $a\in\CA$, we denote
${\rm supp}(a)=\big\{ i\,\big|\,1\le i\le \hat t,\,a_i=a\widehat e_i\ne 0\big\}$, and
\begin{align} \label{eq supp of a}
\textstyle
s_a\!=|{\rm supp}(a)|, ~~~~~~
\CA(a)=\bigoplus\limits_{i\in{\rm supp}(a)} \CA\widehat e_i.
\end{align}

\begin{lemma} \label{lem supp of a}
Assume that $m\geq \mu_q$,
where $\mu_q$ is defined in Eq.\eqref {eq m<mu_q}.
For $(a,b)\in\CA^2=\CA\times \CA$, denote
$$
\CD(\CA,*)_{(a,b)}=
\{\,g\;|\; g\in\CD(\CA,*),\,(a,b)\in C_{1,g} \,\}.
$$
If $\CD(\CA,*)_{(a,b)}\ne\emptyset$,
then ${\rm supp}(a)={\rm supp}(b)$
(equivalently, $\CA(a)=\CA(b)$) and
\begin{align*}
\big|\CD(\CA,*)_{(a,b)}\big|\le
q^{2^r (\hat t-s_a)}.
\end{align*}
\end{lemma}

\begin{proof}
Assume that $g\in\CD(\CA,*)_{(a,b)}$.
Then $(a,b)\in C_{1,g}$, i.e., there is an $u\in\CA$ such that
$(a,b)=u(1,g)$. So $a=u$  and $b=ug=ag$. By Eq.\eqref{eq ab=a_1b_1+...},
$$
  b_i=a_ig_i, ~~~~~~ i=1,\dots,\hat t.
$$
Note that $g_i\in(\CA\widehat e_i)^\times$
(since $g_ig_i^*=-\widehat e_i$),
we see that $b_i\ne 0$ if and only if $a_i\ne 0$.
So ${\rm supp}(a)={\rm supp}(b)$.

For $a=a_1+\dots+a_{\hat t}$ with $a_i\in\CA\widehat e_i=\CA e_i\oplus\CA e_i^*$,
we write $a_i=a_i'+a_i''$ with $a_i'\in\CA e_i$ and $a_i''\in\CA e_i^*$.
It is the same for $b$ and $g$.

Assume that $i\in{\rm supp}(a)$, then $a_i =a_i'+a_i''\ne 0$, hence
one of $a_i'$ and $a_i''$ is nonzero. Suppose that $a_i'\ne 0$,
there is at most one $g_i'\in(\CA e_i)^{\times}$ such that $b_i'=a_i'g_i'$.
By Eq.\eqref{eq aa^*=...}, $g_i''=-(g_i'^{*})^{-1}$ is uniquely determined by $g_i'$.
So there is at most one $g_i\in\CD(\CA\widehat e_i,*)$
such that $b_i=a_ig_i$.

Next assume that $i\notin{\rm supp}(a)$, then $a_i=b_i=0$, hence
any $g_i\in\CD(\CA\widehat e_i,*)$ satisfies that $b_i=a_ig_i$.

By Eq.\eqref{eq D(A,*)=...},
from the above two paragraphs we see that
\begin{align*}
\big|\CD(\CA,*)_{(a,b)}\big|
\leq\prod_{i\notin {\rm supp}(a)}\big|\CD(\CA\widehat e_i,*)\big|
=\big|\CD(\CA,*)\big| \Big/ \!
\prod_{i\in {\rm supp}(a)}\big|\CD(\CA\widehat e_i,*)\big|.
\end{align*}
By Eq.\eqref{eq D A widehat e_i}, 
$$
\prod_{i\in {\rm supp}(a)}\big|\CD(\CA\widehat e_i,*)\big|
=\prod_{i\in {\rm supp}(a)}(q^{2^r}-1)
=(q^{2^r}-1)^{s_a}
$$
Note that $\big|\CD(\CA,*)\big|=(q^{2^r}-1)^{\hat t}$,
see Lemma~\ref{lem |D(A,*)|=}(2).  
Hence
\begin{align*}
 \big|\CD(\CA,*)_{(a,b)}\big|
\leq (q^{2^r}-1)^{\hat t}\big/(q^{2^r}-1)^{s_a}
 = (q^{2^r}-1)^{\hat t-s_a}
\leq (q^{2^r})^{\hat t-s_a}.
\end{align*}
That is,
$\big|\CD(\CA,*)_{(a,b)}\big|\le q^{2^r(\hat t-s_a)}$.
\end{proof}

For $1\le s\le \hat t$, set $\Omega_s$ to be the set of the ``$*$''-stable
ideals of $\CA$ which is a direct sum of $s$ minimal ``$*$''-stable ideals
(i.e., the ``$*$''-stable ideals of dimension $s\cdot 2^{r+1}$),
cf.~Eq.\eqref{eq A=A widehat e}.
Precisely, for $A\in\Omega_s$, there is a unique subset
$\{i_1,\,\dots,\,i_s\}\subseteq\{1,\,\dots,\, \hat t\}$
with $1\le i_1<\dots<i_s\le \hat t$, such that
$A=\CA\widehat e_{i_1}\oplus\dots\oplus\CA\widehat e_{i_s}$.
Let $A\in\Omega_s$. For any $a\in A$, $s_a\le s$ obviously.
We set $\tilde A=\{a\,|\,a\in A,\, s_a=s\}$, and denote
(for $(A\!\times\! A)^{\le\delta}$, cf. Lemma~\ref{lem balance})
\begin{align} \label{eq tilde A}
\textstyle
 (\tilde A\!\times\!\tilde A)^{\le\delta}
 =(\tilde A\!\times\!\tilde A) {\bigcap} (A\!\times\! A)^{\le\delta}
 =\big\{(a,b)\,\big|\, (a,b)\!\in\! \tilde A\!\times\!\tilde A,\,
 \frac{{\rm w}(a,b)}{2\cdot 2^m}\le \delta \big\}.
\end{align}

\begin{lemma} \label{lem D <= delta}
Assume that $m\geq\mu_q$, and that
$\delta\in(0,1-q^{-1})$ and $h_q(\delta)<\frac{1}{4}$.
Denote
$$
 \CD(\CA,*)^{\le\delta}=
  \big\{\, g\,\big|\,g\in\CD(\CA,*),\,\Delta(C_{1,g})\leq\delta\big\}.
$$
Let $\Omega_s$ and $ (\tilde A\times\tilde A)^{\le\delta}$ defined as above.
Then
$$
\CD(\CA,*)^{\le\delta} \;\subseteq\;
\bigcup_{s=1}^{\hat t}\bigcup_{A\in\Omega_{s}}
\bigcup_{(a,b)\in (\tilde A\times\tilde A)^{\le\delta}}  \CD(\CA,*)_{(a,b)}.
$$
\end{lemma}

\begin{proof}
Let $g\in \CD(\CA,*)^{\le\delta}$, there is an $(a,b)\in\CA\times\CA$
such that $(a,b)\in C_{1,g}$ and $\frac{{\rm w}(a,b)}{2\cdot 2^m}\le \delta$,
which implies that $g\in \CD(\CA,*)_{(a,b)}$.
And, by Lemma~\ref{lem supp of a} and Eq.\eqref{eq supp of a},
we get
$\CA(a)=\CA(b)\in\Omega_{s_a}$
and $(a,b)\in(\tilde A\times\tilde A)^{\le\delta}$,
where $A={\CA(a)}$.
\end{proof}

\begin{lemma} \label{lem lem D <= delta}
Keep the assumption as in Lemma~\ref{lem D <= delta}.
Then
\begin{align} \label{eq |D < delta| =}
\big|\CD(\CA,*)^{\le\delta}\big|\,\leq \,
 2^{\hat t}\,q^{2^{m-1} + 4\cdot 2^r \big(h_q(\delta)  - \frac{1}{4}\big) }.
\end{align}
\end{lemma}
\begin{proof}
By Lemma~\ref{lem D <= delta} and Lemma~\ref{lem supp of a},
\begin{align*}
&\big|\CD(\CA,*)^{\le\delta}\big|
~\le~ \sum_{s=1}^{\hat t}~\sum_{A\in\Omega_{s}}
\sum_{(a,b)\in (\tilde A\times\tilde A)^{\le\delta}}  \big|\CD(\CA,*)_{(a,b)}\big| \\
&\le \sum_{s=1}^{\hat t}~\sum_{A\in\Omega_{s}}
\sum_{(a,b)\in (\tilde A\times\tilde A)^{\le\delta}} q^{2^r(\hat t-s)}
=\sum_{s=1}^{\hat t}~\sum_{A\in\Omega_{s}}
|(\tilde A\times\tilde A)^{\le\delta}|\cdot q^{2^r (\hat t-s)}.
\end{align*}
Using Eq.\eqref{eq tilde A} yields
$|(\tilde A\times\tilde A)^{\le\delta}|\le|(A\times A)^{\le\delta}|$.
For $A\in\Omega_s$, $\dim_F(A\times A)=2\cdot s\cdot 2^{r+1}=s\cdot 2^{r+2}$.
By Lemma~\ref{lem balance},
\begin{align*}
|(\tilde A\times\tilde A)^{\le\delta}|\cdot q^{2^r({\hat t}-s)}
&\le |(A\times A)^{\le\delta}|\cdot q^{2^r({\hat t}-s)}
\le q^{s\cdot 2^{r+2} h_q(\delta)}\cdot q^{2^r({\hat t}-s)} \\
&=q^{ 4\cdot s 2^r h_q(\delta) + \hat t\, 2^{r}  - s 2^r}
  = q^{\hat t\, 2^{r} + 4 s 2^r \big(h_q(\delta)  - \frac{1}{4}\big) }.
\end{align*} 
Note that $\hat t=2^{\mu_q-2}$ and $r=m-\mu_q+1$, see Eq.\eqref{eq hat t= r=...}.
So
\begin{align} \label{eq hat t 2^r}
 \hat t\, 2^r=2^{\mu_q-2}\,2^{m-\mu_q+1}=2^{m-1}.
\end{align}
Thus
\begin{align*}
\big|\CD(\CA,*)^{\le\delta}\big|
&\leq \sum_{s=1}^{\hat t}~\sum_{A\in\Omega_{s}}
 q^{2^{m-1}  + 4 s 2^r \big(h_q(\delta)  - \frac{1}{4}\big) }\\
&=\sum_{s=1}^{\hat t} |\Omega_s|\cdot
 q^{2^{m-1} + 4 s 2^r \big(h_q(\delta)  - \frac{1}{4}\big) }.
\end{align*}
Since $h_q(\delta)  - \frac{1}{4}<0$ and $s\ge 1$,
we deduce that
$4 s 2^r \big(h_q(\delta)  - \frac{1}{4}\big)
 \leq 4\cdot 2^r \big(h_q(\delta)  - \frac{1}{4}\big) $. Hence
\begin{align*}
\big|\CD(\CA,*)^{\le\delta}\big|
\le \Big(\sum_{s=1}^{\hat t} |\Omega_s|\Big)\cdot
 q^{2^{m-1} + 4\cdot 2^r \big(h_q(\delta)  - \frac{1}{4}\big) }.
\end{align*}
From the definition of $\Omega_s$,
we get $|\Omega_s|=\binom{\hat t}{s}$, and then
$$\sum_{s=1}^{\hat t} |\Omega_s|=\sum_{s=1}^{\hat t} \binom{\hat t}{s}
= 2^{\hat t}-1\le 2^{\hat t}.$$
The inequality Eq.\eqref{eq |D < delta| =} holds.
\end{proof}

By \cite[Lemma 4.7]{LF22},
if integers $k_i\geq \log_{q}(v)$ for $i=1, \dots, v$,
then
\begin{align} \label{eq LF22 4.7}
\begin{array}{l}
 (q^{k_1}-1) \dots (q^{k_v}-1)\geq q^{k_1+\dots+k_v -2}.
\end{array}
\end{align}

\begin{lemma} \label{D(A,*) >=}
Assume that  $m\geq 2\mu_q$,
where $\mu_q$ is defined in Eq.\eqref {eq m<mu_q}.
Then
$$\big|\CD(\CA,*)\big|\ge q^{-2} q^{2^{m-1}}.$$
\end{lemma}

\begin{proof}
Note that $r=m-\mu_q+1$ and $\hat t=2^{\mu_q-2}$ see Eq.\eqref{eq hat t= r=...}.
From the assumption ``$m\ge 2\mu_q$'', we deduce that $\mu_q-2< m-\mu_q+1$,
and then
$$
 \log_q(\hat t)\le \hat t =  2^{\mu_q-2}\le 2^{m-\mu_q+1}=2^r.
$$
Thus we can quote Eq.\eqref{eq LF22 4.7} and Lemma \ref{lem |D(A,*)|=}(2) to get that
$$
\big|\CD(\CA,*)\big|=(q^{2^r}-1)^{\hat t} \ge q^{\hat t\,2^r -2}.
$$
By Eq.\eqref{eq hat t 2^r}, the inequality of the lemma holds.
\end{proof}

As mentioned in the beginning of this section, to prove Theorem~\ref{int asymptotic},
it is enough to consider the case that $q$ is odd.

\begin{theorem} \label{thm asymptotic}
Assume that $q$ is odd. Let $\delta\in(0,1-q^{-1})$ and $h_q(\delta)<\frac{1}{4}$.
Then there are self-dual $2$-quasi negacyclic codes
$C_1 ,C_2, \dots$ over $F$ (hence the rate ${\rm R}(C_i)=\frac{1}{2}$) such that
the code length $2n_i$ of $C_i$ goes to infinity
and the relative minimum weight $\Delta(C_i)>\delta$ for $i=1,2,\dots$.
\end{theorem}

\begin{proof}
Let $m\geq 2\mu_q$ and $m$ goes to infinity. By Lemma \ref{lem lem D <= delta}
and Lemma~\ref{D(A,*) >=},
\begin{align*}
 \frac{|\CD(\CA,*)^{\le\delta}|}{|\CD(\CA,*)|}
\leq
\frac{2^{\hat t}\,q^{2^{m-1} + 4\cdot 2^r \big(h_q(\delta)  - \frac{1}{4}\big) }}
{q^{-2}q^{2^{m-1}}}
=
2^{\hat t}\,q^2\, q^{ 4\cdot 2^r \big(h_q(\delta)  - \frac{1}{4}\big) }.
\end{align*}
Note that $\hat t=2^{\mu_q-2}$ is independent of the choice of $m$,
and $2^r = 2^{m-\mu_q+1}\to\infty$ while $m\to\infty$, see Eq.\eqref{eq hat t= r=...}.
Since $h_q(\delta)-\frac{1}{4}<0$, 
we get that
\begin{align*}
\lim_{m\to\infty} \frac{|\CD(\CA,*)^{\le\delta}|}{|\CD(\CA,*)|}
\leq
 \lim_{m\to\infty}
2^{\hat t}\,q^2\,
 q^{ 4\cdot 2^r \big(h_q(\delta)  - \frac{1}{4}\big) }
=0.
\end{align*}
Therefore, there are integers $m_1,m_2,\dots$ going to infinity
such that
$$
 \CD(\CA^{(i)},*)^{\le\delta}\,\subsetneqq\, \CD(\CA^{(i)},*),
  ~~ i=1,2,\dots,
$$
where $\CA^{(i)}=F[X]\big/\langle X^{2^{m_i}+1}\rangle$.
We can take
$$
 g^{(i)}\in \CD(\CA^{(i)},*)\setminus \CD(\CA^{(i)},*)^{\le\delta}, ~~~~~i=1,2,\dots.
$$
So that the self-dual $2$-quasi negacyclic codes
$C_i=C_{1,g^{(i)}}$ over $F$ of length $2n_i=2\cdot 2^{m_i}$ goes to infinity
 satisfy that $\Delta(C_i)>\delta$ for all $i=1,2,\dots$.
\end{proof}

\section{Conclusions}\label{Conclusions}

The main purpose of this paper is to investigate
the existence and the asymptotic property of self-dual $2$-quasi negacyclic codes.
When $q$ is even, the self-dual $2$-quasi negacyclic codes over $F$
are just the self-dual $2$-quasi cyclic codes
whose existence and the asymptotic properties are studied in \cite{AOS,L PhD,LF22,LS01,LS03} etc.

In this paper, we primarily concern the case that $q$ is odd.
First, we introduced the operator ``$*$'' on
the quotient algebra $\CR=F[X]/\langle X^{n}+1\rangle$,
which is an algebra automorphism of order $2$,
see Lemma \ref{lem order of *}. 
Specifically, we denote $\CA=F[X]/\langle X^{2^m}+1\rangle$ with $m\geq 1$.
By analysing the behaviour of the primitive idempotents
of $\CA$ 
under the operator ``$*$'',
we obtained the decomposition of $\CA$, see Eq.\eqref{eq A=A widehat e},
which is the key to investigate the existence
and the asymptotic property of self-dual $2$-quasi negacyclic codes.

When $q$ is odd,
we considered the existence of self-dual $2$-quasi negacyclic codes
in two cases: $n$ is odd or $n$ is even. Precisely,
when $n$ is odd, by an algebra isomorphism
$\eta^{(2)}$ (see Eq.\eqref{eq eta 2}) which introduced in \cite{FL24},
we proved that $q$-ary
self-dual $2$-quasi negacyclic codes of length $2n$ exist
if and only if $q\,{\not\equiv}\,-1~({\rm mod}~4)$,
see Theorem \ref{th existence odd}.
When $n$ is even,
we first exhibited an estimation of the number of
the self-dual $2$-quasi negacyclic codes of length $2\cdot2^m$ with ${m\geq 1}$,
see Lemma~\ref{lem |D(A,*)|=},
and then proved the self-dual $2$-quasi negacyclic codes of length~$2n$
 always exist, see Theorem \ref{thm existence even}.
Further, we calculated the number of the
self-dual $2$-quasi negacyclic codes whose relative minimum weight are small,
see Lemma~\ref{lem D <= delta} and Lemma~\ref{lem lem D <= delta}.
In this way, we obtained that
$q$-ary self-dual negacyclic codes are asymptotically good,
see Theorem~\ref{thm asymptotic}.


\end{document}